\documentclass[11pt]{article}
\usepackage[utf8]{inputenc}
\usepackage[margin=1in]{geometry}

 \usepackage{times}
 
 \usepackage{subcaption}
 \usepackage{graphicx}

 \usepackage{times}
 \usepackage{amsfonts,amsmath,amssymb, mathtools, color,float,graphicx,verbatim}
 \usepackage[ruled,vlined,linesnumbered]{algorithm2e}
 \usepackage{hyperref} 
\usepackage{multirow}  
 \usepackage{amsthm,amsfonts,amsmath,color,float,graphicx,verbatim}
 \usepackage{multirow}
 \usepackage{amsmath, color, enumitem}
 \usepackage{framed}
 \usepackage{nicefrac}
 \usepackage{thm-restate}

\theoremstyle{definition}
\newtheorem{definition}{Definition}
\newtheorem{theorem}[definition]{Theorem}
\newtheorem{lemma}[definition]{Lemma}
\newtheorem{cor}[definition]{Corollary}

\begin{document}

\title{Local Access to Sparse Connected Subgraphs Via Edge Sampling}
\author{Rogers Epstein\thanks{Massachusetts Institute of Technology, \href{malito:rogersep@mit.edu}{rogersep@mit.edu}}}


\maketitle
\thispagestyle{empty}

\begin{abstract}
We contribute an approach to the problem of locally computing sparse connected subgraphs of dense graphs. In this setting, given an edge in a connected graph $G = (V, E)$, an algorithm locally decides its membership in a sparse connected subgraph $G^* = (V, E^*)$, where $E^* \subseteq E$ and $|E^*| = o(|E|)$. Such an approach to subgraph construction is useful when dealing with massive graphs, where reading in the graph's full network description is impractical.

While most prior results in this area require assumptions on $G$ or that $|E'| \le (1+\epsilon)|V|$ for some $\epsilon > 0$, we relax these assumptions. Given a general graph and a parameter $T$, we provide membership queries to a subgraph with $O(|V|T)$ edges using $\widetilde{O}(|E|/T)$ probes. This is the first algorithm to work on general graphs and allow for a tradeoff between its probe complexity and the number of edges in the resulting subgraph.

We achieve this result with ideas motivated from edge sparsification techniques that were previously unused in this problem. We believe these techniques will motivate new algorithms for this problem and related ones. Additionally, we describe an efficient method to access any node's neighbor set in a sparsified version of $G$ where each edge is deleted with some i.i.d. probability.
\end{abstract}

\newpage
\pagenumbering{arabic} 
\section{Introduction}

Many real-world applications of graph algorithms apply to massive inputs, from social networks to core internet infrastructure. Given that many algorithms must be run frequently, classical models of computation become immensely slow and inefficient on these large graphs. The study of sublinear algorithms aims to find fast procedures that only look at a small fraction of the input while minimizing the error in the result.

In computational graph problems, many sublinear algorithms aim to simulate query access to some function of the input graph. One class of algorithms that fall in this category are Local Computation Algorithms, or LCAs, as defined in \cite{DBLP:journals/corr/abs-1104-1377}. An LCA aims to ``maintain'' some global solution while only performing sublinear probes on the input for each time the algorithm is queried. Namely, an LCA may, for example, allow a user to query the color of a given node such that the result is consistent with a valid $k$-coloring. Other examples of LCAs determine if a given vertex is in some maximal independent set that is generated using the procedure's internal randomness.

We are interested in algorithms that ``maintain'' some subgraph $H$; given any edge in the original graph, they look at a sublinear number of edges to determine whether or not the given edge is in $H$. Such an algorithm aims to answer consistently among all possible query inputs, as if it actually has the global solution in memory and were able to access $H$ directly. While LCAs strive for probe complexities that are sublinear in the number of vertices, being sublinear in the number of edges suffices for our purposes.

In particular if we are given some $G = (V, E)$, we strive to provide query access to some sparse, connected subgraph $G^* = (V, E^*)$ where $E^* \subseteq E$. As such, our algorithm $\mathcal{A}$ is a function $\mathcal{A}: E \rightarrow \{accept, reject\}$, where $e \in E^*$ if and only if $\mathcal{A}(e) = accept$. We also want each call to $\mathcal{A}$ to make some sublinear number of probes (i.e. $o(|E|)$) to the original graph structure.

This motivates the following definition:

\begin{definition}\label{lscg}
\textbf{(Local Algorithms for Sparse Connected Graphs)} An algorithm $\mathcal{A}$ is a local sparse connected graph (LSCG) algorithm yielding a subgraph with $f(|V|)$ edges if given probe access to $G = (V, E)$, $\mathcal{A}$ tests membership of $e$ in some subgraph of $G$, $G^* = (V, E^*)$, such that the following conditions hold with high probability (over the internal coin flips of $\mathcal{A}$):

$1.$ $G^*$ is connected

$2.$ $|E^*| \le f(|V|)$

\end{definition}

The ``local'' description of these algorithms comes from them requiring a sublinear number of probes with respect to the size of the original graph $G$. Additionally, these algorithms are ``sparse'' because we enforce that $f(|V|) = o(|E|)$. To implement such an algorithm, we assume we have a public source of (unbounded) randomness.

While performing a breadth first search from any node could yield an optimal $G^*$ (only $|V| - 1$ edges), the probe complexity of such an algorithm would be linear in the original number of edges, and thus undesirable in our setting. On the other hand, an algorithm that always returns $accept$ would have zero probe complexity but $|E^*| = |E|$, which is also undesirable. With most of the prior work having focused on minimizing $|E'|$, we hope to find some balance within this tradeoff of probe complexity and the number of edges in the resulting subgraph.

\subsection{Prior Work}

The study of this problem was initiated in \cite{SpanningGraphs}, where the authors provide LSCG algorithms that yield subgraphs with at most $(1+\epsilon)n$ edges using $O(1)$ probes on bounded-degree graphs with low expansion properties. They also design an algorithm that uses $\widetilde{O}(\sqrt{n})$\footnote{The notation $\widetilde{O}(g(n))$ is equivalent to $O(g(n)(\log n)^a)$, for some constant $a$.} probes on bounded-degree graphs with high expansion. They show this probe complexity is essentially tight by demonstrating that for general bounded-degree graphs, $\Omega(\sqrt{n})$ probes are necessary per edge query.

Subsequent results continued to aim for subgraphs with at most $(1+\epsilon)n$ edges, but also hoped for a subgraph of small \emph{stretch}. The measure of stretch of a subgraph is equivalent to the maximum distance between two nodes in the subgraph that share an edge in the original graph. To obtain such subgraphs, these works also focused on special degree-bounded classes of input graphs (\cite{levi2015constructing}, \cite{minorfree}, \cite{canwe}, \cite{LenzenL17}). The techniques developed by Lenzen-Levi were expanded upon in \cite{spannerLCAs} to give the first LSCG algorithm with a less strict edge bound. In particular, given a degree upper bound of $\Delta$, the authors give an algorithm that constructs a connected subgraph with $\widetilde{O}(n^{1+1/k})$ edges using $\widetilde{O}(\Delta^4n^{2/3})$ probes per edge query.

Importantly, this was also the first work to give an LSCG algorithm that works on general graphs. For either $r \in \{2, 3\}$, the authors give an algorithm that yields a subgraph with $\widetilde{O}(n^{1+1/r})$ edges using $\widetilde{O}(n^{1-1/(2r)})$ probes per edge query.

While all of these works use different analyses, they largely rely on the same strategy: identify clusters of nodes in the graph, generate sparse spanning subgraphs within those clusters, and then connect these clusters using few edges. Our paper is the first to use an entirely different tactic.

\subsection{Overview of Results and Techniques}

In this work we provide and analyze an LSCG algorithm such that when given a valid\footnote{As will be justified in Section \ref{simple-anal}, we consider a $T$ such that $T = \omega(\log^2 n)$ and $T = \widetilde{o}(m)$} input parameter $T$, yields a subgraph with $f(|V|) = O(|V|T)$ edges. Additionally, its probe complexity is $\widetilde{O}(|E|/T)$. As will be discussed in Section \ref{anal}, this improves over prior results on $\Delta$ bounded-degree graphs for $\Delta = \Omega(n^{1/9})$. Additionally, this is the first result that works for general graphs and allows for a spectrum of upper bounds on the number of edges in the resulting subgraph.

The main approach of this algorithm is to attempt to measure a metric of the ``connectivity'' of an edge, and keep each edge with some probability that is a function of this metric. Specifically, this metric is the ``strong connectivity'' of an edge as defined by Benczúr-Karger in \cite{StrongConnKarger}, and is repeated here in Definition \ref{strongdef}. Our technique at large is described in depth in Section \ref{mainch}.

The strong connectivity of an edge is approximated by a test that accesses a random, sparsified copy $G' = (V, E')$ of the input graph $G$, where each edge is kept independently with some equal probability\footnote{Such a subgraph is called a \emph{skeleton}.}. In particular, this test aims to fully explore some connected component in $G'$, and doing so with $O(|E'|)$ probes instead of $O(|E|)$ is a nontrivial task. As this is an interesting result in of itself, we dedicate Section \ref{sparsifying_graphs} to describing how to efficiently provide this local access. In Section \ref{anal} we analyze our algorithm for special cases and compare it to past results. Finally, we give discuss relevant open questions related to this problem and our techniques in Section \ref{future}.

\section{Preliminaries}\label{ch2:pre}

\subsection{Notation and the Model}\label{model}

We consider general graphs $G = (V, E)$ that are simple and undirected. Here, $V$ is the set of vertices of the graph, and $E$ is the set of edges connecting vertices. We define $n = |V|$ and $m = |E|$, and we assume that each vertex has $v \in V$ has a unique ID, over which there is some full ordering. We denote $N_G(u)$ to be the neighbor set of $u$ in a graph $G$, which we may write as $N(u)$ if $G$ is clear from context.

We assume probe access to a modified version of the input graph's incidence-list representation. Specifically, we are able to make constant time \textsc{degree}, \textsc{neighbor} and \textsc{adjacency} queries, where $\textsc{neighbor}(u, i)$ gives the $i$th neighbor of $u$, and $\textsc{adjacency}(u, v)$ gives $i$ if $v$ is the $i$th neighbor of $u$, and $\bot$ if they are not adjacent. Note that these are the same queries allowed in \cite{spannerLCAs}.

Our main goal is to construct an LSCG algorithm by locally accessing (with high probability) an unweighted version of an \emph{$\epsilon$-sparsification} of $G$, where $\epsilon > 0$.

\begin{definition}\label{sparsification}
An \emph{$\epsilon$-sparsification} of a graph $G = (V, E)$ is a weighted subgraph $G' = (V, E')$ of $G$ where $E' \subseteq E$ with a weight function $w: E' \rightarrow \mathbb{R}$, such that for all cuts $(S, V \backslash S)$, the weighted value of this cut in $G'$ is within a $(1\pm \epsilon)$ multiplicative factor of the value of the corresponding cut in $G$.
\end{definition}

\subsection{Sparsification Results}\label{ch2:b-k}

The core idea of the approach in this thesis is based on definitions and theorems initiated in \cite{StrongConnKarger}, along with its preliminary version, \cite{BenczurK96}. In these papers the authors consider a strengthening of the notion of $k$-connected components\footnote{These are components that have a min cut at least $k$.}:

\begin{definition}[\cite{BenczurK96}]
A \emph{$k$-strong component} is a maximal $k$-connected vertex-induced subgraph.
\end{definition}

Here, maximal means no vertices can be added to the component that would maintain its $k$-connectedness. Recall that a subgraph is vertex-induced if it contains exactly the edges in the original graph whose endpoints are both in the subgraph. Using $k$-strong components, \cite{BenczurK96} defines the \emph{strong connectivity} of an edge:

\begin{definition}[\cite{BenczurK96}]\label{strongdef}
Given an edge $e$ in a graph $G = (V, E)$ where $e \in E$, the \emph{strong connectivity} $s_e$ of $e$ is the largest value $k$ such that $e$ is contained in a $k$-strong component.
\end{definition}

One motivation for this definition comes from the following theorem of Karger, which deals with graphs with a known minimum cut:

\begin{theorem}[\cite{Karger:1994:RSC:195058.195422}]\label{mincut-thm}
Given parameters $n$ and $d$, define $\lambda'_{\epsilon} = \frac{3}{\epsilon^2}(d+2)\log n$. In a graph with $n$ vertices and min-cut $c$, independently keeping each edge with probability $p = \min(1, \lambda'_{\epsilon}/c)$ and giving it weight $1/p$ yields an $\epsilon$-sparsification, with error probability $O(n^{-d})$.
\end{theorem}

\begin{figure}[h]\label{fig:sparse}
\centering
\includegraphics[width=\textwidth]{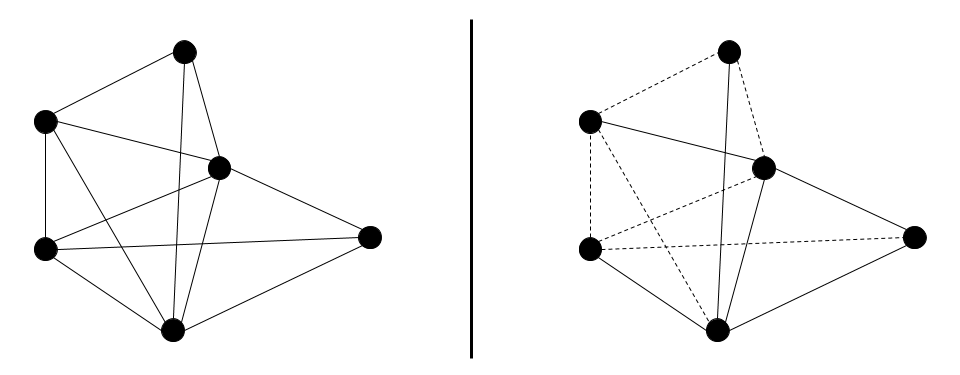}
\caption{Given a graph $G$ on the left, we can produce a skeleton $G'$ on the right by independently at random keeping each edge according to some probability $p$. The dotted lines represent an edge that was deleted with probability $1-p$.}
\end{figure}

Building on this result, Benczúr-Karger show that if edges can be sampled as a function of their strong-connectivity instead of the global min-cut, then a stronger statement is possible. In particular, the following theorem reduces the number of edges in the subgraph resulting from the sparsification procedure, assuming $c = \widetilde{o}(m/n)$:

\begin{theorem}[\cite{BenczurK96}]\label{sparse-thm}
Given parameters $n$ and $d$, define $\lambda_{\epsilon} = \frac{16}{\epsilon^2}(d+2)\log n$. In a graph with $n$ vertices, independently keeping each edge $e$ with probability $p_e = \min(1, \lambda_{\epsilon}/s_e)$ and assigning it weight $1/p_e$ yields an $\epsilon$-sparsification, with error probability at most $n^{-d}$.
\end{theorem}

It should be recognized that these procedures use full independence. As we are mostly concerned with our algorithm's probe and runtime complexity, we assume quick access to arbitrarily many independent bits. If needed, these bits can be generated efficiently using pseudorandom generators, as described in \cite{prg}.

One barrier to the usefulness of Theorem \ref{sparse-thm} is the difficulty in calculating the exact strong connectivity of edges. The following corollary suggests an easier approach:

\begin{cor}\label{approx}
Suppose that for each edge $e$, we have some approximation of strong connectivity $\hat{s_e}$ such that $s_e/\alpha \le \hat{s_e} \le s_e$. Then, sampling each edge with probability $\lambda_{\epsilon}/\hat{s_e}$ (and using the corresponding reweighting scheme from Theorem \ref{sparse-thm}), results in a graph where all cut values are preserved to within a multiplicative factor of $(1\pm \epsilon)$. Additionally, the graph will have $O(\alpha n\log n)$ edges with high probability.
\end{cor}

In \cite{StrongConnKarger}, the authors utilize Nagamochi-Ibaraki certificates (\cite{Nagamochi:1992:CEC}, \cite{journals/algorithmica/NagamochiI92}) to achieve a $2$-approximation in $\widetilde{O}(m)$ time. While is not known how to locally simulate this approximation scheme, a local version of this procedure would imply a LSCG algorithm. In proving the above theorems, \cite{StrongConnKarger} also showed several lemmas that will be useful for this work:

\begin{lemma}[\cite{StrongConnKarger}]\label{edgebnd}
For any $t \in \mathbb{N}$, there are at most $t(n-1)$ edges $e \in E$ such that $s_e \le t$.
\end{lemma}

\begin{lemma}[\cite{StrongConnKarger}]\label{recipsum}
$\sum_{e \in E} 1/s_e \le n-1$.
\end{lemma}

The approach to sparsification exemplified by Theorem \ref{sparse-thm} has been widely adopted for many problems related to the original goal of approximating the min-cut (\cite{example1}, \cite{StreamingKarger}). In particular, \cite{StreamingKarger} applies their techniques in the streaming setting, and motivates some of our applications of the above theorems.

Additionally, many works have shown that similar sampling processes work by using different edge-connectivity metrics. Such measures include edge conductance, effective resistance, $s$-$t$ connectivity, and Nagamochi-Ibaraki (NI) indices (\cite{conductance}, \cite{resistance1}, \cite{resistance2}, \cite{resistance3}).

Our results also build on the problem of providing local access to large random graphs through partial sampling. In particular we are inspired by \cite{largegoldreich} and \cite{local-random}. The latter work gives a way to query subsequent neighbors of any node in a random instance of $G(n, p)$. These results will be described more in depth, and expanded upon, in Section \ref{sparsifying_graphs}.

\section{The Sublinear Algorithm}\label{mainch}

In this section we present our main local algorithm. Given access to a large input graph, it locally computes membership to some sparse, connected subgraph. More specifically, we present an LSCG algorithm $\mathcal{A}$ such that on an input $e \in E$, $\mathcal{A}(e)$ answers $accept$ if and only if $e \in H \subseteq G$. Here, $H$ is connected and has at most $O(n(T+\log^2n))$ edges for some parameter $T$. Also, the probe complexity of $\mathcal{A}$ is at most $\widetilde{O}(m/T)$.

\subsection{Main Ideas}

According to Corollary \ref{approx}, if one can approximate the Benczúr-Karger strong connectivities of an edge to within a factor of $\alpha$, then sampling edges accordingly yields an $\epsilon$-sparsification (as described in from Definition \ref{sparsification}). This subgraph will have $O(\alpha n\log n)$ edges. In our algorithm, we give an $\alpha=O(\log n)$-approximation scheme using sublinear many probes.

In addition, we set $\epsilon = 1/2$ for the remainder of this paper, as constructing $1/2$-sparsifications suffice for our purposes. According to Theorems \ref{mincut-thm} and \ref{sparse-thm}, we set the relevant parameters as follows:

\begin{definition}\label{lambda}
$\lambda = \lambda_{1/2} = 64(d+2)\log |V|$
\end{definition}

\begin{definition}\label{lambdaprime}
$\lambda' = \lambda'_{1/2} = 12(d+2)\log |V|$
\end{definition}

Since we have $\epsilon < 1$, note that any $\epsilon$-sparsification of a connected graph is connected itself. This is the case because for any cut in $G$ that has value $c$, the corresponding cut in $G'$ has value at least $c(1-\epsilon) > 0$. Thus, no cut with positive value becomes zero, and the whole graph remains connected. If we ignore edge weights, positive cuts remain positive, so the unweighted version of a such an $\epsilon$-sparsification is a connected subgraph.

Our approximation scheme for edge strong connectivity works by testing ``guesses''. In the next section, we will describe a test such that for a guess $g$, the test rejects with high probability if $g \ge 2\lambda' s_e$ (recall $\lambda'$ from Definition \ref{lambdaprime}) and accepts with high probability if $g \le s_e$. Then we can run this test on powers of $2$ and with high probability, the largest accepted guess $g^*$ has the property that $\frac{s_e}{2} \le g^* \le 2\lambda' s_e$. This statement is proved in Lemmas \ref{smallguess} and \ref{bigguess}. Finally, we can set $\hat{s_e} = \frac{g^*}{2\lambda'}$ to get an $\alpha = 4\lambda' = O(\log n)$-approximation of $s_e$.\footnote{Here we use the notion of approximation scheme as described in Corollary \ref{approx}.}

As we will see, running this test for low-valued guesses requires the largest number of probes of any part of this algorithm. As such, we set some threshold $T$, and only make guesses above this value. If all our guesses are rejected, we simply keep the queried edge for our overall algorithm. Otherwise, we are able to compute some $\hat{s_e}$, and keep the queried edge with probability $\lambda/\hat{s_e}$.

\subsection{Testing a Guess for Strong Connectivity}

Suppose we have a guess $g$ for the strong connectivity $s_e$ of an edge $e = (u, v)$. We propose a test for this $g$ that has the following steps:

First, choose one of its endpoints arbitrarily, say $u$, and initialize $S = V$. Then, run the following procedure $O(\log n)$ times: keep each edge in the vertex-induced subgraph induced by $S$, $G_t$, with probability $p = \lambda'/g$, and output ``reject'' if $v$ is not reachable from $u$ in the resulting graph\footnote{Here we use $\lambda' = 12(d+2)\log|S|$, not $12(d+2)\log n$}. For the next procedure, redefine $S$ to be the set of vertices reachable from $u$ in this subgraph. If we do not reject in any of these iterations, output ``$accept$.'' Algorithm \ref{tester} is described more precisely below.

\begin{algorithm}
	\SetKwInOut{Input}{Input}
	\SetKwInOut{Output}{Output}
	\Input{$G = (V, E)$, $e \in E$, guess $g$}
	\Output{accept or reject}
	\DontPrintSemicolon
	$(u, v) \gets e$ \; 
	$S \gets V$ \;
	\For{$\lceil\log_{3/2}(n)\rceil$ rounds}{
	Let $G_t$ be the subgraph of $G$ induced by $S$ \;
	Construct $G'$ by sampling each edge of $G_t$ with probability $p = \lambda'/g$ \;
	$S \gets$ set of nodes reachable from $u$ in $G'$ \;
	\If{$v \notin S$}{
	Output reject and abort.
	}
	}
	Output accept.
	\caption{\textsc{Edge Strong Connectivity Tester}}
	\label{tester}
\end{algorithm}

By the final setting of $S$, this test distinguishes two important kinds of guesses, as shown in the following two lemmas:

\begin{lemma}\label{smallguess}
If $g \le s_e$, then $v \in S$ with high probability.
\end{lemma}

\begin{proof}
Recall from Definition \ref{strongdef} that if an edge $e = (u, v)$ has strong connectivity $s_e$, then it is contained in some vertex-induced subgraph $H \subseteq G$, where $H$ has min-cut $s_e$. Note that the existence of $H$ also implies that for all $e' \in H$, $s_{e'} \ge s_e$. Our test of guess $g$ keeps each edge of the graph with probability $p = \frac{\lambda'}{g}$, and then tests if $u$ and $v$ are in the same connected component. Since $g \le s_e$, we know $p \ge \frac{\lambda'}{s_e}$. By Theorem \ref{mincut-thm}, $H$ is connected w.h.p., and thus so are $u$ and $v$. Thus, $v \in S$ with high probability for each iteration. Since we run $O(\log n)$ iterations, $v$ is in the final $S$ with error probability $O(n^{-d}\log n)$, using a union bound.
\end{proof}

\begin{lemma}\label{bigguess}
If $g \ge 2\lambda' s_e$ for some $e=(u,v)$, then $v \notin S$ with high probability, where $\lambda'$ is defined in Definition \ref{lambdaprime}.
\end{lemma}
\begin{proof}

The subgraph of $G$ induced by edges with strong connectivity $\ge k = s_e + 1$ involves some number of node-disjoint strongly connected components (\cite{StrongConnKarger}). Suppose there are $n'$ such $k$-strong components, where $u$ and $v$ are in different such components because $s_e < k$. We now consider the graph where each such component is its own supernode. Note that edges in this graph correspond to inter-component edges in the edge-induced subgraph of $G$, since they have strong connectivity $\le s_e < k$. Note that if a vertex in $G$ is not incident to any edge with strong connectivity at least $k$, then that vertex will be its own supernode in the resulting graph. By Lemma \ref{edgebnd} the number of the edges in this new graph (i.e. the inter-component edges) is at most $s_e(n'-1)$. Thus, if we sample edges with probability $p = \lambda'/g \le 1/(2s_e)$, then at most $\frac{2}{3}(n'-1)$ of these edges are kept w.h.p. using a Chernoff bound.

Thus, at least $1/3$ of the supernodes become isolated from $u$'s supernode in the sparsified graph. As we run subsequent rounds in Algorithm \ref{tester}, we can consider this supernode graph on the vertices that are still connected to $u$. In each round, either $u$'s supernode becomes isolated or at least $1/3$ of the remaining supernodes become isolated (w.h.p.). Thus, in $\lceil\log_{3/2}(n)\rceil$ rounds, $u$'s supernode must be isolated with high probability. In particular, the error probability is at most $O(n^{-d}\log n)$ by taking a union bound on the error probability from Theorem \ref{mincut-thm}.
\end{proof}

\subsection{Algorithm Description and Complexity Analysis}\label{simple-anal}

There are a few subtleties of this algorithm that remain. For one, running these tests requires fast access to a random, sparsified skeleton of the original graph. This is a nontrivial task, but is achieved in Section \ref{sparsifying_graphs}. For our analysis here, it suffices to know that exploring $u$'s connected component by running a breadth first search can be done in $\widetilde{O}(n+mp)$ probes, where each edge is added to the skeleton graph with probability $p$. Recall that in a test with guess $g$, we sample edges with probability $p = \lambda'/g$. So, the test with the largest probe complexity is that with the smallest guess $g$, which will be at least $T$. Thus, each test runs with probe complexity at most $\widetilde{O}(n+m/T)$.

As mentioned above, we can simplify the number of guesses we have on a particular edge by simply guessing powers of $2$. Note that the strong connectivity of any edge in a simple graph must be between $1$ and $n-1$ inclusive. Since we are only testing guesses that are at least $T$, it is sufficient to have $k$ such guesses, where $2^{k+1} \ge n/T$. This inequality is satisfied for $k = O(\log(n/T)) = O(\log n)$. This bound can be improved for specific edges by noting that for any edge $e = (u, v)$, $s_e \le \min(deg(u), deg(v))$. This is because for any vertex-induced subgraph containing $e$, its min-cut is always at most the cut created by separating $e$'s lower-degree endpoint from the rest of the subgraph. So if we knew that the original input graph had some polynomial upper degree bound $\Delta$, then we would only need at most $O(\log(\Delta/T) = O(\log \Delta)$, though this will not improve the asymptotic bound. For more analysis on the probe complexity for special kinds of graphs, see Section \ref{anal}. We recap our LSCG algorithm below in Algorithm \ref{mainalg}.

\begin{algorithm}
	\SetKwInOut{Input}{Input}
	\SetKwInOut{Output}{Output}
	\Input{$G = (V, E)$, $e = (u, v)\in E$, $T$}
	\Output{accept or reject}
	\DontPrintSemicolon
	$\hat{s_e} \gets 0$ \;
	$g \gets \min(deg(u), deg(v))$ \;
	\While{$g > T$}{
	Run Algorithm \ref{tester} on $G, e, g$ \;
	\If{Accepted}{
	$\hat{s_e} \gets g/(2\lambda')$ \;
	Break \textbf{for} loop \;
	}
	$g \gets g/2$ \;
	}
	\If{$\hat{s_e} = 0$}{
	$\hat{s_e} = T$ \;
	}
	Accept $e$ with probability $\lambda/\hat{s_e}$, otherwise reject.
	\caption{\textsc{Main LSCG Algorithm}}
	\label{mainalg}
\end{algorithm}

With this knowledge, we can upper bound the total probe complexity from running all these tests by $\sum_{k=\log_2 T}^{\log_2 n} O(n + \frac{m\log n}{2^{k}}) \le O(n\log n + m\frac{\log n}{T}) = O((n+m/T)\log n)$. For choices of $T$ such that $T = \widetilde{\omega}(m/n)$, this is $O((m/T)\log n)$. Note that this bound on $T$ is needed to guarantee $\widetilde{o}(m)$ edges in the subgraph, so we assume this restriction for the remainder of the paper.

Additionally, we can bound the number of edges in the final subgraph: firstly, if the input edge $e$ has $s_e \le T$, then $e$ will be accepted. From Lemma \ref{edgebnd}, there are at most $T(n-1) = O(nT)$ such edges. Additionally, the number of edges $e$ with $s_e > T$ that are accepted is at most the number of edges added by the original Benczúr-Karger scheme using an $\alpha$-approximation, which is $O(\alpha n\log n)$ with high probability. Since we have an $O(\log n)$-approximation, this adds $O(n\log^2(n))$ edges to our final connected subgraph. Suppose we consider choices of $T = \omega(\log^2(n))$, which is required to get a probe complexity that is sublinear in $m$ up to polylogarithmic factors. In fact we assume this restriction for the remainder of the paper. Then, the number of edges in the final subgraph is at most $O(nT + n\log^2(n)) = O(nT)$, with high probability.

\section{Accessing Skeletons of General Graphs}\label{sparsifying_graphs}

In Algorithm \ref{mainalg}, specifically through its calls to Algorithm \ref{tester}, we require the ability to search through many skeletons $G'$ of our original graph $G$. In particular, these are subgraphs where the edges of $G$ are kept independently with some specified probability $p$ (alternatively, the edges are independently deleted with probability $1-p$). We require the ability to determine if two specified nodes are path-connected within each skeleton. We test for this by running a breadth first search. Though if we make additional assumptions on the input graph, there may be a more efficient test, which we will discuss in more depth in Section \ref{future}. To implement this breadth first search, we give a way to efficiently access any node's neighbors within $G'$.

In particular, we give an efficient implementation of a data structure that allows us to access all the neighbors of $u$ in $G'$, namely $N_{G'}(u)$ (recall this definition from Section \ref{model}). The goal is to allow this access in $\widetilde{O}(|N_{G'}(u)|)$ queries, which we accomplish by making $|N_{G'}(u)| + 1$ \textsc{next-neighbor} queries on $u$. If we think of $N_{G'}(u)$ as an ordered set, our first \textsc{next-neighbor} query on $u$ gives the first element, the $i$th query gives the $i$th element, and after all the neighbors have been returned, all future queries yield $\bot$. Thus, an efficient implementation of \textsc{next-neighbor} queries allows us to access $N_{G'}(u)$ efficiently; specifically we hope to achieve this access in $\widetilde{O}(|N_{G'}(u)|)$ probes to $G$. This allows us to run a breadth first search on $G'$ in $O(\sum_{u \in V} |N_{G'}(u)|)$ queries, which is $\widetilde{O}(mp)$ with high probability.

There are some immediate barriers to doing this efficiently. If $u \in N_{G'}(v)$, for example, it must always be the case that $v \in N_{G'}(u)$. However, simply having a coin (with weight $p$) for each possible edge and checking the result of each coin flip would require $\Omega(|N_G(u)|)$ total work to determine $N_{G'}(u)$, which is undesirable.

When the input graph is the complete graph $K_n$, it is known to be possible to query $N_{G'}(u)$ via \textsc{next-neighbor} queries using $\widetilde{O}(|N_{G'}(u)|)$ total probes (\cite{local-random})\footnote{They also implements \textsc{random-neighbor} and \textsc{degree} queries with $\widetilde{O}(1)$ probe complexity, though these details are not required for our purposes.}. 

\subsection{Fast \textsc{next-neighbor} Queries on Sparsified $K_n$}\label{bry}

First, we restate and rephrase some of the relevant techniques of \cite{local-random}. In the next section, we present our modifications of their techniques to suit our more general purposes.

Consider a complete graph $K_n$ such that $V(K_n) = [n]$ and for all nodes $u \in [n]$, $N_{K_n}(u)$ is in increasing order. Suppose we aim to provide query access to a skeleton $K'$, which is constructed by keeping each edge of $K_n$ with i.i.d. probability $p$. We do so by maintaining an auxiliary data structure that mimics ``filling in'' the adjacency matrix $\textbf{A}$. One can think of the entries $\textbf{A}[u][v]$ as being $1$ if some \textsc{next-neighbor} query implies that $(u, v) \in K'$, $0$ if the results of the queries definitively implies that $(u,v) \notin K'$, and $\phi$ otherwise to mark that it is to be decided. Mimicking access to $\textbf{A}$ is done by maintaining two types of quantities for each node $u$: $\textbf{last}[u]$ and $P_u$, which we will describe in the next two paragraphs.

$\textbf{last}[u]$ is a pointer to the last neighbor of $u$ that is returned by a \textsc{next-neighbor} query. It is initialized to $-1$, since no neighbors have been returned at this point. Once \textsc{next-neighbor} has been called on $u$ at least $|N_{K'}(u)|$ times, $\textbf{last}[u]$ will continue to be $u$'s final neighbor.

$P_u$ is the ordered set of known neighbors of $u$, given all the prior outputs of \textsc{next-neighbor} queries. It is initialized to the empty set $\{\}$, since nothing is known about $K'$ before any queries are made. Once \textsc{next-neighbor} has been called on $u$ at least $|N_{K'}(u)|$ times, $P_u$ will be exactly $N_{K'}(u)$.

Both of these parts of our data structure will be updated as we make more \textsc{next-neighbor} queries on our sparsified skeleton of $K_n$. They can then be used to infer the entries of the adjacency matrix. Specifically, $\textbf{A}[u][v]$ is $1$ if $u \in P_v$ or $v \in P_u$\footnote{We maintain the invariant that $u \in P_v$ if and only if $v \in P_u$.}. Otherwise, it is $0$ if $u < \textbf{last}[v]$ or $v < \textbf{last}[u]$. If neither of these cases hold, we say $\textbf{A}[u][v] = \phi$.

These inferences are correct as $(u, v) \in K'$ if $u \in P_v$ or $v \in P_u$. We also know that $(u, v) \notin K'$ once a \textsc{next-neighbor} query on $u$ yields a node that comes \textit{after} $v$, meaning $\textbf{last}[u] > v$ (or similarly $\textbf{last}[v] > u$).

With this data structure, \cite{local-random} is able to efficiently simulate the distribution $F(u, a, b)$ of $u$'s first neighbor in $K'$ whose ID is between nodes $a$ and $b$, exclusive. This distribution is conditioned on the current known state of $\textbf{A}$ (i.e. the $\textbf{last}[u]$'s and $P_u$'s), and takes advantage of the fact that the location of the next neighbor is distributed according to a hyper-geometric distribution. Its exact implementation is omitted in this paper. However, Biswas et. al. note that sampling from this distribution can be done in \emph{constant} time. For more details on this, refer to \cite{local-random}.

With this information, the authors implement \textsc{next-neighbor} queries using the steps of Algorithm \ref{biswas}. For an input vertex $u$, we start at its last accessed neighbor $v = \textbf{last}[u]$, and look until its next known neighbor $w_u$. In this interval we sample its next appearing neighbor, which is given by the hypergeometric distribution $F(u, v, w_u)$ until we find a new vertex $v$ that was previously not known to neighbor $u$\footnote{Note that by construction, our sampled $v$'s will always have the property that $\textbf{last}[u] \ge v$, so what remains to check is if $\textbf{last}[v] \ge u$. If this latter inequality holds, then $(u,v)$ was already a determined edge of $K'$.}. \cite{local-random} show that with high probability, we need to sample from $F$ at most $O(\log n)$ times to find such a vertex. Then, we update the data structure and output the discovered neighbor $v$, if there exists one.

\begin{algorithm}
	\SetKwInOut{Input}{Input}
	\SetKwInOut{Output}{Output}
	\Input{$G = (V, E)$, $u \in V$}
	\Output{accept or reject}
	\DontPrintSemicolon
	$v \gets \textbf{last}[u]$ \;
	$w_u \gets \min\{(P_u \cap (v, n]) \cup \{n+1\}\}$ \;
	\While{$v \neq w_u$ \text{and} $\textbf{last}[v] \ge u$}{
	Sample $v \sim F(u, v, w_u)$ \;
	}
	\If{$v \neq w_u$}{
	$P_u \gets P_u \cup \{v\}$ \;
	$P_v \gets P_v \cup \{u\}$
	}
	$\textbf{last}[u] \gets v$ \;
	\leIf{$v \le n$}{\Return $v$ \;}{\Return $\bot$}
	\caption{$\textsc{next-neighbor}(u)$ on $K_n$ (\cite{local-random})}
	\label{biswas}
\end{algorithm}

Each line of Algorithm \ref{biswas} requires $O(1)$ work; since the while loop is known to execute $O(\log n)$ times with high probability, a \textsc{next-neighbor} query runs in $O(\log n)$ time. Thus, we can recover $N_{K'}(u)$ with $O(|N_{K'}(u)|\log n)$ time and probe complexity.

\subsection{Adapting this Approach for General Graphs}

While \cite{local-random} provides and proves the ability to efficiently make $\textsc{next-neighbor}$ queries on a random skeleton $K'$ of $K_n$ (as described in Section \ref{bry}), it is possible to extend this approach for skeletons $G'$ of general graphs, assuming our basic types of queries from Section \ref{model}.

Implementing this extension requires a slight tweak of the definitions of the data structure. Instead of storing the ID of the last outputted neighbor of $u$ in $\textbf{last}[u]$, we store its index in $N(u)$. Similarly, instead of storing the known neighbors of $u$ in $P_u$, we maintain a sorted list of the indices in $N(u)$ of $u$'s known neighbors in $G'$. Additionally, we modify the implementation of $F(u, a, b)$ to represent the distribution of $u$'s first neighbor in $K'$ whose \emph{index} in $N(u)$ is between values $a$ and $b$, exclusive. These will be needed to access the adjacency arrays containing $u$'s neighbors.

Generalizing as such gives us Algorithm \ref{extension} below. Note that we use each of the graph queries described in Section \ref{model}. \textsc{degree} queries are needed to determine the largest possible index that corresponds to one of $u$'s neighbors. \textsc{neighbor} queries are needed to use these indices to quickly access the relevant vertices. Finally, \textsc{adjacency} queries help maintain the invariant that $i \in P_u$ and $v$ being $u$'s $i$th neighbor is equivalent to $j \in P_v$ and $u$ being $v$'s $j$th neighbor.

\begin{algorithm}[!ht]
	\SetKwInOut{Input}{Input}
	\SetKwInOut{Output}{Output}
	\Input{$G = (V, E)$, $u \in V$}
	\Output{accept or reject}
	\DontPrintSemicolon
	$v_{\text{index}} \gets \textbf{last}[u]$ \;
	$w_{\text{index}} \gets \min\{(P_u \cap (v_{\text{index}}, \textsc{degree}(u)]) \cup \{\textsc{degree}(u)+1\}\}$ \;
	\While{$v_{\text{index}} \neq w_{\text{index}}$ and $\textbf{last}[v_{\text{index}}] \ge u$}{
	Sample $v_{\text{index}} \sim F(u, v_{\text{index}}, w_{\text{index}})$
	}
	$v \gets \textsc{neighbor}(u, v_{\text{index}})$ \;
	\If{$v_{\text{index}} \neq w_{\text{index}}$}{
	$P_u \gets P_u \cup \{v_{\text{index}}\}$ \;
	$P_v \gets P_v \cup \{\textsc{adjacency}(u, v)\}$
	}
	$\textbf{last}[u] \gets v_{\text{index}}$ \;
	\leIf{$v_{\text{index}} \le \textsc{degree}(u)$}{\Return $v$ \;}{\Return $\bot$}
	\caption{$\textsc{next-neighbor}(u)$}
	\label{extension}
\end{algorithm}

Each line of Algorithm \ref{extension} requires $O(1)$ work; since the while loop is known to execute $O(\log n)$ times with high probability, this adapted \textsc{next-neighbor} query runs in $O(\log n)$ time. Thus, we can recover $N_{G'}(u)$ with $O(|N_{G'}(u)|\log n)$ time and probe complexity, using $|N_{G'}(u)| + 1$ \textsc{next-neighbor} queries.

\section{Further Analysis and Comparison of Results}\label{anal}

In this Section, we further our analysis of Algorithm \ref{mainalg}. There are special classes of graphs where the probe complexity of our algorithm can be improved significantly. Additionally, there are certain cases of the LSCG problem where we improve on known results, or provide results where none were known previously.

\subsection{Further Performance Analysis}

One useful observation about our analysis is that Algorithm \ref{mainalg} on an edge $e$ will only run Algorithm \ref{tester} as a subroutine on a guess $g$ if $g/2 \ge s_e$, with high probability. This is the case because the main algorithm multiplies rejected guesses by $1/2$, and w.h.p. Algorithm \ref{tester} accepts any $g \le s_e$. The overall probe complexity is largest when testing the smallest guess, and the probe complexity of a guess $g$ is $\widetilde{O}(m/g)$. Since we only test guesses $g \ge T$, the probe complexity of Algorithm \ref{mainalg} on an edge $e$ can be refined to $\widetilde{O}(\min(m/T, m/s_e))$.

Without any assumptions about the graph or about the queried edge, the probe complexity could be as large as $\widetilde{O}(m/T)$. However, if our graph has min-cut $c \ge T$, then it must be true that for all edges $e \in E$, $s_e \ge c$. Thus, we can tighten our probe complexity bound to $\widetilde{O}(m/c)$ in such graphs. In particular, on strongly-connected and dense graphs ($c = \Theta(n)$, $m = \Theta(n^2)$), our algorithm achieves a probe complexity of just $\widetilde{O}(n)$. Also, since all edges have strong connectivity at least $T$, our algorithm finds an $O(\log n)$-approximation for the strong connectivity of every edge. Thus, the number of edges in the final subgraph is at most $O(n\log^2 n)$ by Corollary \ref{approx}.

These statements are true without the algorithm having to know or compute the value of $c$; we just have to choose $T \le c$. This is an improvement over Theorem \ref{mincut-thm}, where knowing the value of $c$ in advance is required. Additionally, since having a min-cut of $c$ implies $m = \Omega(nc)$\footnote{Since the degree of each node must be at least $c$.}, our result improves on the number of edges in the sparsified subgraph.

We can also reason about the expected probe complexity of a random edge:

\begin{theorem}
The average probe complexity when this algorithm is queried on a random edge is $\widetilde{O}(nT)$.
\end{theorem}

\begin{proof}
We compute this quantity as the average probe complexity over all edges, which is $\frac{\sum_{e \in E} \widetilde{O}(\min(m/T, m/s_e))}{m}$. We can upper bound the numerator of this expression in a similar manner to how we bounded the number of edges in the resulting subgraph at the end of section \ref{simple-anal}. By Lemma \ref{edgebnd}, at most $nT$ edges have the property that $\min(m/T, m/s_e) = m/T$. Also from Lemma \ref{recipsum}, we know that $\sum_{e \in E} 1/s_e \le n - 1$. So, 

$$\frac{\sum_{e \in E} \widetilde{O}(\min(m/T, m/s_e))}{m} = \widetilde{O}(\frac{mnT + mn}{m}) = \widetilde{O}(nT)$$
\end{proof}

Lastly, it should be noted that the asymptotics of our algorithm's time complexity matches that of its probe complexity. This is because the most computationally intensive step of Algorithm \ref{mainalg} is in running a breadth first search, in which the probe and time complexities match.

\subsection{Comparison of Bounds to Related Works}

As this methodology for locally constructing sparse connected subgraphs differs greatly from the approaches used thus far, it is useful to compare the successes of these techniques. In particular, we will mainly compare probe complexities of this result with those achieved in \cite{spannerLCAs} for the LSCG problem, which is the only other result that applies to accessing connected subgraphs with more than $(1+\epsilon)n$ edges.

Their main results were as follows:

\begin{enumerate}
    \item For graphs with degree upper bound $\Delta$, one can construct a subgraph with $\widetilde{O}(n^{1+1/k})$ edges using $\widetilde{O}(\Delta^4 n^{2/3})$ probes per edge.
    
    \item For general graphs and $r \in \{2, 3\}$, one can construct a subgraph with $\widetilde{O}(n^{1+1/r})$ edges in $\widetilde{O}(n^{1-1/(2r)})$ probes per edge.
\end{enumerate}

Our algorithm is not an improvement over the latter result, as we only generate a subgraph with $\widetilde{O}(n^{3/2})$ edges at the cost of $\widetilde{O}(m/\sqrt{n})$ probes per edge, which is worse for the practical case of $m = \Omega(n^{3/2})$. Similarly, our new results give worse probe complexities for subgraphs with $\widetilde{O}(n^{4/3})$ edges. That said, those results were previously the only ones known to work on general graphs. Our approach improves on this problem by allowing for a subgraph with any upper bound (in particular, one that is $\widetilde{O}(n^{4/3})$) on the number of edges, given a general graph.

For degree-bounded graphs, there are parameters where our algorithm improves on the probe complexity of the first result. In such graphs, we can construct a subgraph with $\widetilde{O}(n^{1+1/k})$ edges in $\widetilde{O}(m/n^{1/k}) = \widetilde{O}(n^{1-1/k}\Delta)$ probes, setting our threshold parameter $T = n^{1/k}$. This is better than $\widetilde{O}(\Delta^4n^{2/3})$ for graphs with degree maximum $\Delta$ of $\Omega(n^{1/9})$\footnote{Specifically, $\Delta = \widetilde{\omega}(n^{1/9-1/(3k)})$}.

Both of the above prior results also come with stretch guarantees, whereas the techniques in this paper do not immediately allow for any such guarantee. However, we discuss a possible remediation of this in section \ref{future}.

\section{Future Work and Open Questions}\label{future}

While the setup of the LSCG problem is basic and motivated, there is still much to learn. Most of the works in this area have been done using similar techniques and on bounded degree graphs. Also, the goal is often to give access to a subgraph with at most $f(n) = (1+\epsilon)n$ edges for some $\epsilon > 0$. Only recently has there been work on graphs of unbounded degree, or with less strict bounds on the number of edges in the subgraph (\cite{spannerLCAs}).

While this work considers looser bounds on the final number of edges, it would be interesting to look at bounds that are stricter but not as much as $(1+\epsilon)n$. For example, are there techniques that perform asymptotically better than that those yielding subgraphs with at most $(1+\epsilon)n$ edges, but allow as many as $2n$? $3n$? $cn$ for some $c \in \mathbb{N}$? \cite{spannerLCAs} proved a lower bound of $\Omega(\sqrt{n})$ probes on constant degree-bounded graphs to yield a subgraph with $(1+\epsilon)n$ edges, and this bound is achieved in many special cases of such graphs (\cite{SpanningGraphs}). Seeing that this lower bound has not been matched yet for general constant degree-bounded graphs, it might be useful to relax the problem by allowing for more edges.

Additionally, for $\Delta$ degree-bounded graphs the only result for subgraphs with $\omega(n)$ edges requires $\Omega(n^{2/3})$ probes (\cite{spannerLCAs}). Since the best known lower bound is $\Omega(n^{1/2})$ for $\Delta = O(n^{1/2})$, there is a decent gap in our understanding within the general bounded-degree graph regime. This remains true when we relax the degree bound constraint. The optimal probe complexity for an algorithm to access a connected subgraph with $\widetilde{O}(n^{1+\epsilon})$ edges, for $\epsilon \ge 0$, is still unknown in general. It is reasonable to think that this gap between the known upper and lower bounds can be closed, in part because the prior works also aim for a subgraph of low stretch (\cite{levi2015constructing}, \cite{minorfree}, \cite{LenzenL17}, \cite{spannerLCAs}). It would be interesting to see what improvements can be made on the probe complexity of LSCG algorithms if constraints on the stretch are not imposed.

As discussed towards the end of Section \ref{ch2:b-k}, the Benczúr-Karger sampling scheme also works with several metrics $m_e$ of different edges, achieving a subgraph with $\widetilde{O}(n)$ edges where we keep each edge with probability $\frac{O(\log n)}{m_e}$. These metrics include edge conductance, Nagamochi-Ibaraki indices, and other measures of connectivity. Finding a way to locally approximate these measures is a compelling problem in and of itself, but doing so would also yield an LSCG algorithm with a similar approach to ours.

As for our provided algorithm, it is likely that the analysis can be optimized further. To find out if $u$ and $v$ are connected in the skeleton graphs of Algorithm \ref{tester}, we run a breadth first search, which may look at all the nodes and edges of $u$'s connected component. Under certain assumptions\footnote{For example, if the graph is rapid mixing.}, it might be possible to check if $v$ is in $u$'s connected component much more efficiently. One conceivable improvement of the analysis might be to show that on some strong class of input graphs and with reasonable probability, the skeleton graphs of Algorithm \ref{tester} satisfy one of those desired assumptions, for some range of strong connectivity guesses $g$. In turn, this would allow for a more efficient version of Algorithm \ref{tester}. 

Another improvement might come in bounding the size of the vertex-induced subgraph $H\ni e$ with min-cut $s_e$ guaranteed by the definition of strong connectivity (Definition \ref{strongdef}). Having such an understanding would upper bound the required depth of the BFS of Algorithm \ref{tester}, and give stretch guarantees to the final subgraph. A similar breakthrough could also be achieved by the relaxed hope that for some parameters $\alpha, \beta$, if an edge has strong connectivity $s_e$, then there exists a vertex-induced subgraph $H'\ni e$ with min-cut $\alpha s_e$ and of size $f(\beta)$.

A similar possible approach could be to study a local version of strong connectivity:

\begin{definition}
Given an edge $e$ in a graph $G = (V, E)$ where $e \in E$ and some distance parameter $r$, let the \emph{local strong connectivity} $s_{e, r}$ of $e$ be the largest value $k$ such that $e$ is contained in a $k$-strong component $H$, where all nodes in $H$ are within distance $r$ of an endpoint of $e$.
\end{definition}

It would be useful if on some $\Delta$ bounded-degree class of graphs, ``most'' edges $e$ have the property that $s_e/s_{e,r} \le \alpha$. Then, by calculating the local strong connectivity and sampling according to Theorem \ref{sparse-thm}, we would yield a connected subgraph with $O(\alpha n\log n)$ edges. Additionally, running a breadth first search in a bounded degree graph with depth at most $r$ uses only $\widetilde{O}(\Delta^r)$ probes. Even if $O(\alpha n\log n)$ edges had the property that $s_e/s_{e,r} > \alpha$, sampling these edges accordingly would add a small number of edges to the $1/2$-sparsification, notably maintaining the asymptotics on the number of edges in the final subgraph. Determining what $\alpha$ values are possible for given values of $r$ seem directly related to the study of Extremal Graph Theory. An overview of the area can be found in \cite{fredi2013history}.

Finally, while the lower bounds of probe complexity achieved by \cite{spannerLCAs} for the LSCG problem are essentially tight for certain, constant degree-bounded graphs (\cite{SpanningGraphs}), this is not true for denser graphs. In this regime, there is a large gap between the known upper and lower bounds of how many probes are needed to achieve LSCG algorithms. Any improvement on these bounds for the probe complexities of LSCG algorithms in this setting, or the problem at large, is still of great interest.

\newpage
\bibliographystyle{alpha}
\bibliography{main}
\end{document}